\documentclass[%
superscriptaddress,
 amsmath,amssymb,
 aps,
 pra,
floatfix,
twocolumn
]{revtex4-1}
\usepackage[utf8]{inputenc}
\usepackage[T1]{fontenc}
\usepackage{babel}
\usepackage{graphicx}
\usepackage{amsfonts}
\usepackage{amssymb}
\usepackage{amsmath}
\usepackage{amsthm}
\usepackage{mathtools}
\usepackage{physics}
\usepackage[normalem]{ulem}
\usepackage{blindtext}

\DeclareMathOperator{\arcosh}{arcosh}

\usepackage{braket}
\renewcommand\bra[1]{{\langle{#1}|}}
\renewcommand\ket[1]{{|{#1}\rangle}}
\usepackage{multirow}
\usepackage{tabularx}
\usepackage{dsfont}
\usepackage{bm}
\usepackage{xcolor}

\newtheorem{theorem}{Theorem}
\newtheorem{proposition}[theorem]{Proposition}
\newtheorem{corollary}[theorem]{Corollary}

\theoremstyle{definition}
\newtheorem{example}{Example}

\usepackage{natbib}
\setcitestyle{square,numbers}
\usepackage[colorlinks]{hyperref}
\usepackage[format=plain,justification=centerlast]{caption}
\usepackage[format=plain,justification=centerlast]{subcaption}

\begin{document}

\author{Tomasz Linowski}
\affiliation{International Centre for Theory of Quantum Technologies, University of Gdansk, 80-308 Gda{\'n}sk, Poland}
\email[Corresponding author: ]{t.linowski95@gmail.com}

\author{Konrad Schlichtholz}
\affiliation{International Centre for Theory of Quantum Technologies, University of Gdansk, 80-308 Gda{\'n}sk, Poland}

\author{{\L}ukasz Rudnicki}
\affiliation{International Centre for Theory of Quantum Technologies, University of Gdansk, 80-308 Gda{\'n}sk, Poland}
\affiliation{Center for Theoretical Physics, Polish Academy of Sciences, 02-668 Warszawa, Poland}

\title{Formal relation between Pegg-Barnett and Paul quantum phase frameworks}

\date{\today}

\begin{abstract}
The problem of defining a hermitian quantum phase operator is nearly as old as quantum mechanics itself. Throughout the years, a number of solutions was proposed, ranging from abstract operator formalisms to phase-space methods. In this work, we make an explicit connection between two of the most prominent approaches, by proving that the probability distribution of phase in the Paul formalism follows exactly from the Pegg-Barnett formalism by combining the latter with the quantum limited amplifier channel. Our findings suggest that the Paul framework may be viewed as a semi-classical limit of the Pegg-Barnett approach.
\end{abstract}

\maketitle

\section{Introduction}

In the history of quantum mechanics, few problems received as much attention as the problem of definition and measurement of the phase of quantum electromagnetic field. Since the early failed attempt by Dirac \cite{phase_operator_Dirac_1927}, a plethora of solutions was proposed, including phase-space approaches \cite{quantum_phase_Wheeler_1985,quantum_phase_phase-space_methods_Leonhardt_1993}, as well as operator formalisms by Susskind and Glogower \cite{quantum_phase_operator_Susskind_1964}, Garrison and Wong \cite{quantum_phase_operator_Garrison_1970}, Paul \cite{quantum_phase_operator_Paul_1974}, L\'{e}vy-Leblond \cite{quantum_phase_operator_Levy-Leblond_1976}, Popov and Yarunin \cite{quantum_phase_operator_Popov_Yarunin_1991}, and, finally, Pegg and Barnett \cite{quantum_phase_operator_Pegg_1988,quantum_phase_operator_Pegg_1989}.

Significant interest was devoted especially to the Pegg-Barnett formalism, in which the obstacles standing in the way of a well-defined quantum phase operator are overcome by reducing the problem to a finite dimension. After its discovery, the formalism quickly gave rise to an alternative derivation \cite{quantum_phase_operator_PB_alternative_derivation_Luis_1993} and extension \cite{quantum_phase_difference_operator_Luis_1993}, among others \cite{quantum_phase_operator_PB_Lahti_2002,quantum_phase_operator_Honarasa_2009,quantum_phase_PB_Perez_Leija_2016,quantum_phase_Ramos-Prieto_2020}. Nowadays, the Pegg-Barnett formalism is used, e.g., to investigate phase properties of various non-classical phenomena, including photon antibunching in the case of photon addition and substraction \cite{addition_subtraction_Malpani19,addition_subtraction_Malpani20}, phase-number squeezing in atom field interactions \cite{Number_phase_squuezing_Faghihi} and nonlinear squeezed states \cite{nonlinear_squeezed_Vamegh}.

In this work, we focus on the formal aspects of the Pegg-Barnett framework and relate it to the Paul formalism, a ``competing'' solution to the quantum phase problem, notable especially for its close connection to phase-space approaches and a clear experimental interpetation \cite{Paul_phase_experiment_Freyberger_1993,FREYBERGER199341}. More precisely, we prove that the probability distribution in the Paul formalism can be exactly obtained from its counterpart in the Pegg-Barnett formalism if the Pegg-Barnett phase operator is combined with the quantum limited amplifier channel. Furthermore, due to the amplifier's association with classicality \cite{quantum_phase_Q_amplification_Schleich_1992} (discussed in more detail below), we interpret our result as the Paul framework being a semi-classical limit of the Pegg-Barnett formalism. In this way, we bridge the two approaches mathematically and physically.


Interestingly, our findings are not the first to apply the quantum limited amplifier to the Paul formalism, albeit the context is different. In Ref. \cite{quantum_phase_Q_amplification_Schleich_1992}, it was shown that the phase distribution of a quantum state in the Paul framework can be realized experimentally through the amplified state. In Ref. \cite{Paul_formalism_QLA_Lalovi_1998}, the Paul framework was proved to be the only quantum phase description consistent with the Glauber model of amplification and the natural expectation that large-amplitude coherent states should have a well-defined phase.



The article is organized as follows. In Sec. \ref{sec:phase_formalisms}, we briefly summarize the two discussed quantum phase formalisms. In Sec. \ref{sec:QLA}, we introduce our main tool: the quantum limited amplifier channel. In Sec. \ref{sec:main_result}, we state and derive our main result. An in-depth discussion of this result is provided in Sec.~\ref{sec:discussion}, with explicit examples in Sec.~\ref{sec:examples}. We conclude in Sec.~\ref{sec:summary}.

\section{Paul and Pegg-Barnett phase formalisms}
\label{sec:phase_formalisms}
As already stated, our main subject of interest concerns the Paul and Pegg-Barnett formalisms, and the connection between them through the quantum limited amplifier channel. Let us briefly introduce and discuss the two phase formalisms. For a detailed review, see, e.g. \cite{quantum_phase_operator_review_Lynch_1995,quantum_phase_operator_review_Barnett_2007}.

\subsection{Paul formalism}
Back in 1974, Paul considered the following family of operators  \cite{quantum_phase_operator_Paul_1974}
\begin{align} \label{eq:Paul_operators}
\begin{split}
    \hat{E}_k & \coloneqq 
        \int \frac{d^2\alpha}{\pi} \, 
        \left(\frac{\alpha}{|\alpha|}\right)^k \ket{\alpha}\bra{\alpha}, \\
    \hat{E}_{-k} & \coloneqq
        \int \frac{d^2\alpha}{\pi} \, 
        \left(\frac{\alpha^*}{|\alpha|}\right)^k \ket{\alpha}\bra{\alpha} = \hat{E}_k^\dag,
\end{split}
\end{align}
where $k\in\mathbb{N}$ and
\begin{align} \label{eq:coherent_state}
\begin{split}
    \ket{\alpha} = \sum_{n=0}^\infty \alpha_n\ket{n},
        \quad \alpha_n = e^{-|\alpha|^2/2}\frac{\alpha^n}{\sqrt{n!}}
\end{split}
\end{align}
is a coherent state with amplitude $\alpha\in\mathbb{C}$. We discuss the case of a single mode, for which $\hat a$ is the annihilation operator (so that $\hat a\ket{\alpha}=\alpha\ket{\alpha}$), and $\ket{n}$ for $n=0,\ldots,\infty$ denotes the Fock basis.

The operators (\ref{eq:Paul_operators}) correspond to the classical quantities $e^{\pm ik\phi}$. This is most easily seen by setting $\alpha = r e^{i\phi}$ and computing the expectation values on an arbitrary state $\hat{\rho}$, which leads to
\begin{align} \label{eq:Paul_operators_expectation}
\begin{split}
    \braket{\hat{E}_k}_{\hat{\rho}} =
        \int_{0}^{2 \pi} d\phi \, e^{ik\phi} \int_{0}^{\infty} \frac{dr}{\pi} \, r \, 
        Q_{\hat{\rho}}\left(re^{i\phi}\right),
\end{split}
\end{align}
and analogously for $\hat{E}_{-k}$. Here,
\begin{align} \label{eq:Husimi_Q_distribution}
    Q_{\hat{\rho}}(\alpha) \coloneqq \bra{\alpha}\hat{\rho}\ket{\alpha},
\end{align}
is the Husimi Q quasiprobability distribution \cite{Q_representation}. Due to the properties of the Husimi function, the rightmost integral in eq. (\ref{eq:Paul_operators_expectation}) is positive for all $\phi$ and, when integrated over $\phi$ from $0$ to $2\pi$, yields one. For this reason, it can be regarded as the probability distribution of $\phi$ in the formalism:
\begin{align} \label{eq:Paul_p}
      P_{\,\textnormal{Paul}}(\phi|\hat{\rho}) \coloneqq \int_{0}^{\infty} \frac{dr}{\pi} \, r \, Q_{\hat{\rho}}\left(re^{i\phi}\right).
\end{align}

It is useful to compare the Paul operators with the Glauber-Sudarshan P representation \cite{P_representation_Glauber} of an arbitrary operator $\hat{X}$, defined through
\begin{align} \label{eq:Glauber_P_distribution}
    \hat{X} = \int \frac{d^2\alpha}{\pi} \, P_{\hat{X}}(\alpha) \, \ket{\alpha}\bra{\alpha}.
\end{align}
We can see that the Paul operators (\ref{eq:Paul_operators}) are essentially operators whose P distribution is equal to the $k$-th powers of the quantity $e^{i\phi}$. 

As a natural generalization, in this paper we consider operators whose P representation is rendered by any complex-valued, bounded function $f$ of $e^{i\phi}$, i.e.
\begin{align} \label{eq:Paul_operator_f}
    \hat{\phi}_{\textnormal{Paul}}[f] \coloneqq 
        \int \frac{d^2\alpha}{\pi} \, f\left(\frac{\alpha}{|\alpha|}\right) \ket{\alpha}\bra{\alpha}.
\end{align}
Clearly, any such operator has similar properties to the original Paul operators. In particular, its expectation value reads
\begin{align} \label{eq:Paul_expectation}
    \braket{\hat{\phi}_{\textnormal{Paul}}[f]}_{\hat{\rho}}  
        = \int_{0}^{2 \pi} d\phi \, f\left(e^{i\phi}\right) \, P_{\,\textnormal{Paul}}(\phi|\hat{\rho}),
\end{align}
with $P_{\,\textnormal{Paul}}(\phi|\hat{\rho})$ being the same probability distribution as in eq. (\ref{eq:Paul_p}). 


One of the main strengths of the Paul formalism is its close association with experimental phase detection, such as homodyne measurements using the eight-port interferometer \cite{Paul_phase_experiment_Freyberger_1993,FREYBERGER199341}. In the strong local oscillatory regime of such an experiment, the phase difference between an arbitrary bosonic state and a reference coherent state is given precisely by the Paul probability distribution (\ref{eq:Paul_p}).


\subsection{Pegg-Barnett formalism}
\label{sec:PB}
Introduced in 1988 \cite{quantum_phase_operator_Pegg_1988} and developed further during the subsequent years, the Pegg-Barnett formalism is built upon a family of $s+1$ ``number-phase states''
\begin{align} \label{eq:PB_theta_m_state}
    \ket{\theta_{t,s}} \coloneqq \frac{1}{\sqrt{s+1}}\sum_{n=0}^{s} e^{i n \theta_{t,s}} \ket{n},
\end{align}
where
\begin{align} \label{eq:PB_theta_m}
    \theta_{t,s} \coloneqq \frac{2\pi t}{s+1}, \qquad t\in\{0,1,\ldots,s\}.
\end{align}
Typically, an arbitrary reference phase $\theta_0$ is added to the definition of $\theta_{t,s}$. Indeed, from a practical point of view the phase itself is not well-defined and measurements must be made relative to an auxiliary state. For example, see  \cite{quantum_phase_difference_operator_Luis_1993} where an operator measuring the phase difference between two systems was considered. Here, we are concerned with the relation between the Pegg-Barnett and Paul formalisms and not the formalisms themselves. For consistency with the Paul framework, we therefore take the liberty to set the reference phase $\theta_0$ to zero. 

For finite $s$, the number-phase states form an orthonormal basis of the $(s+1)$-dimensional Hilbert space being the $s$-photon subspace of the single-mode Fock space. Hence, the Pegg-Barnett phase operator
\begin{align} \label{eq:PB_operator}
    \hat{\phi}_{\textnormal{PB}}^{(s)} \coloneqq \sum_{t=0}^s \theta_{t,s} 
        \ket{\theta_{t,s}}\bra{\theta_{t,s}},
\end{align}
is hermitian. By considering a formal power series of this operator, we can associate a Pegg-Barnett operator to any complex-valued, bounded function $f$ of the phase exponential:
\begin{align} \label{eq:PB_operator_f}
    \hat{\phi}_{\textnormal{PB}}^{(s)}[f] \coloneqq \sum_{t=0}^s f\left(e^{i\theta_{t,s}}\right)     
        \ket{\theta_{t,s}}\bra{\theta_{t,s}}.
\end{align}
Computing its expectation value on state $\hat{\rho}$, we find
\begin{align} \label{eq:PB_expectation}
    \braket{\hat{\phi}_{\textnormal{PB}}^{(s)}[f]}_{\hat{\rho}} = \sum_{t=0}^s f\left(e^{i\theta_{t,s}}\right)
        \bra{\theta_{t,s}}\hat{\rho}\ket{\theta_{t,s}}.
\end{align}
Since $f$ is arbitrary, we conclude that the probability that the state's phase is equal to $\theta_{t,s}$ is therefore given by
\begin{equation} \label{eq:PB_dist}
    \bra{\theta_{t,s}}\hat{\rho}\ket{\theta_{t,s}}.
\end{equation}

Note that the results so far depend on the auxiliary dimension $s$. This is resolved by taking the limit $s\to\infty$. As the limit is considered, the summation over $t$ in eq. (\ref{eq:PB_expectation}) may be replaced by an integral, so that \cite{Barnett_optics}
\begin{align} \label{eq:PB_expectation_continuous}
    \braket{\hat{\phi}_{\textnormal{PB}}^{(s)}[f]}_{\hat{\rho}} = 
    \int_0^{2\pi} d\phi \, f\left(e^{i\phi}\right) P_{\textnormal{PB}}^{(s)}(\phi|\hat{\rho}).
\end{align}
Here,
\begin{equation} \label{eq:PB_dist_continous}
    P_{\textnormal{PB}}^{(s)}(\phi|\hat{\rho})
        \coloneqq \frac{s+1}{2\pi}
        \bra{\phi_s}\hat{\rho}\ket{\phi_s},
\end{equation}
where
\begin{align} \label{eq:PB_phi_state}
    \ket{\phi_s} \coloneqq \frac{1}{\sqrt{s+1}}\sum_{n=0}^{s} e^{i n \phi} \ket{n}
\end{align}
are the continuous counterparts to the discrete phase states (\ref{eq:PB_theta_m_state}). We remark that in the very limit $s=\infty$, the continuous phase states coincide with those in the Susskind-Glogower formalism \cite{quantum_phase_operator_Susskind_1964}.

For normalizable states, the limit $s\to \infty$ of eq. (\ref{eq:PB_expectation_continuous}) can be computed under the integral:
\begin{align} \label{eq:PB_expectation_continuous_limit}
   \lim_{s\to\infty}\braket{\hat{\phi}_{\textnormal{PB}}^{(s)}[f]}_{\hat{\rho}} 
        =\int_0^{2\pi} d\phi \, f\left(e^{i\phi}\right) 
         \lim_{s\to\infty}P_{\textnormal{PB}}^{(s)}(\phi|\hat{\rho}).
\end{align}
In such cases, the formalism has a well-defined probability distribution in the limit of infinite dimension:
\begin{align} \label{eq:PB_dist_continuous_limit}
    \lim_{s\to\infty}P_{\textnormal{PB}}^{(s)}(\phi|\hat{\rho})
\end{align}
and the expectation values can be computed after taking the limit. 

However, as Barnett himself points out \cite{Barnett_optics}, there exist states for which the order of integration and limit cannot be exchanged, meaning that the formula (\ref{eq:PB_expectation_continuous_limit}) is not always valid. Then, the probability distribution does not exist in the limit $s\to\infty$ and the expectation values have to be computed through either of the following expressions:
\begin{align} \label{eq:PB_expectation_continuous_limit_general}
\begin{split}
   \lim_{s\to\infty}\braket{\hat{\phi}_{\textnormal{PB}}^{(s)}[f]}_{\hat{\rho}} 
        &=\lim_{s\to\infty} \sum_{t=0}^s f\left(e^{i\theta_{t,s}}\right)
        \bra{\theta_{t,s}}\hat{\rho}\ket{\theta_{t,s}}\\
        &=\lim_{s\to\infty}\int_0^{2\pi} d\phi \, f\left(e^{i\phi}\right) 
         P_{\textnormal{PB}}^{(s)}(\phi|\hat{\rho}),
\end{split}
\end{align}
where we stress that in the discussed singular cases the limit in the bottom line has to be performed after integration. The postulate that, in general, the expectation values should be computed first, and only then the limit $s\to\infty$ should be taken, is a key feature of the Pegg-Barnett formalism.

\section{Quantum limited amplifier}
\label{sec:QLA}
To make the connection between the Paul and Pegg-Barnett frameworks, we will use the quantum limited amplifier channel, or QLA in short.

The action of the (one-mode) QLA channel of an arbitrary strength $\kappa\geqslant 1$ on state $\hat{\rho}$ is defined~as~\cite{Quantum_limited_amplifier_Mollow_1967,De_Palma_2017}
\begin{align} \label{eq:QLA}
    \mathcal{A}_\kappa(\hat{\rho})\coloneqq
        \Tr_B\left[\hat{U}_\kappa(\hat{\rho}\otimes\ket{0}\bra{0})\hat{U}^\dag_\kappa\right],
\end{align}
where
\begin{align} \label{eq:squeezing_operator}
    \hat{U}_\kappa\coloneqq
        \exp\left[\arcosh\sqrt{\kappa}(\hat{a}^\dag\hat{b}^\dag-\hat{a}\hat{b})\right]
\end{align}
is the two-mode squeezing operator. Here $\hat{b}$ is the annihilation operator associated with the ancillary system traced out in eq. (\ref{eq:QLA}). The case $\kappa=1$ corresponds to the identity channel.

From the physical point of view, QLA may be viewed as the process of pumping particles into the system. Because of its properties, it is sometimes viewed as a tool for making a quantum state more ``classical'' \cite{quantum_phase_Q_amplification_Schleich_1992}. In particular, it was shown that the Glauber-Sudarshan P quasiprobability distribution of an infinitely amplified state is always non-negative \cite{quantum_phase_Q_amplification_Schleich_1992}, a quality that is associated only with semi-classical states~\cite{beam_splitter_classicality_Brunelli_2015}.

The action of the amplifier on a state can be calculated explicitly in the number basis. Substituting the convenient decomposition \cite{Barnett_optics} of the squeezing operator (\ref{eq:squeezing_operator})
\begin{align} \label{eq:squeezing_operator_normal_ordering}
\begin{split}
    \hat{U}_\kappa=\hat{r}_{+,\kappa}^\dag
    \exp\left[-\ln\sqrt{\kappa}\left(\hat{a}^\dag\hat{a}+\hat{b}^\dag\hat{b}+1\right)\right]
    \hat{r}_{-,\kappa},
\end{split}
\end{align}
where $\hat{r}_{\pm,\kappa}\coloneqq\exp\left[\pm\sqrt{\frac{\kappa-1}{\kappa}}\hat{a}\hat{b}\right]$, into the definition (\ref{eq:QLA}), we eventually obtain
\begin{align} \label{eq:QLA_standard_basis}
\begin{split}
    \mathcal{A}_\kappa(\hat{\rho})=&\:\frac{1}{\kappa}\sum_{\substack{j=0}}^\infty
        \left(\frac{\kappa-1}{\kappa}\right)^{j}
        \sum_{\substack{m,n=0}}^\infty
        \rho_{mn}\frac{1}{\sqrt{\kappa}^{m+n}}\\
    &\sqrt{\binom{j+m}{j}\binom{j+n}{j}}\ket{j+m}\bra{j+n}.
\end{split}
\end{align}
where $\rho_{mn}\equiv\bra{m}\hat{\rho}\ket{n}$.

The QLA channel possesses the semi-group property, according to which
\begin{align} \label{eq:QLA_semigroup}
    \mathcal{A}_x\left[\mathcal{A}_y(\hat{\rho})\right] = \mathcal{A}_{xy}(\hat{\rho})
\end{align}
for any $x,y\geqslant 1$. Finally, we remark that occasionally, for increased readability in subscripts, we will denote QLA by
\begin{align} \label{eq:QLA_alternative}
    \mathcal{A}(\kappa,\hat{\rho})\equiv\mathcal{A}_\kappa(\hat{\rho}).
\end{align}

\section{Main results}
\label{sec:main_result}
We come back to the Pegg-Barnett formalism. From the practical point of view, to obtain the final, $s$-independent measurement of phase, one has to perform a number of measurements of finite-dimensional Pegg-Barnett operators given by different values of $s$. For a large enough number of measurement results obtained for large enough $s$, one can predict the limiting measurement result for $s\to\infty$. 

Let us now imagine that in this setup, before measuring the $s$-dimensional Pegg-Barnett operator, we first apply to the system the QLA channel $\mathcal{A}_{1+s\epsilon}$, $\epsilon\ll 1$. In other words, before measuring the $s$-dimensional Pegg-Barnett operator, we prepare the system in the state $\mathcal{A}_{1+s\epsilon}(\hat{\rho})$. Again, for a large enough number of measurement results obtained for large enough $s$, one can predict the limiting measurement result $s\to\infty$. Our main claim is that, in the limit of vanishing amplification strength, $\epsilon \to 0$, the results obtained from this procedure are indistinguishable from analogous results obtained from the Paul formalism for an unamplified state.

We are now in the position to state our main result. For clarity, we present in the form of a proposition.

\begin{proposition} \label{th:main_result}
The Paul probability distribution (\ref{eq:Paul_p}) can be obtained from the Pegg-Barnett continuous probability distribution (\ref{eq:PB_dist_continous}) through the quantum limited amplifier as:
\begin{align} \label{eq:main_result_prob}
    P_{\,\textnormal{Paul}}(\phi|\hat{\rho})
        = \lim_{\epsilon\to 0}\lim_{s\to\infty}
        P_{\,\textnormal{PB}}^{(s)}[\phi|\mathcal{A}_{1+s\epsilon}(\hat{\rho})].
\end{align}
\end{proposition}

\noindent Before we prove the proposition, let us make two important remarks. 

Firstly, we stress that it is crucial that the order of limits on the r.h.s. of eqs (\ref{eq:main_result_prob}) cannot be changed. If we took the limits in the opposite way, we would obtain no amplification at all, since $\mathcal{A}_1$ is the identity channel. Consequently, in the process described at the beginning of this section, we would be performing the ordinary Pegg-Barnett measurement, which is obviously different from the Paul measurement.

Secondly, in Proposition \ref{th:main_result}, we called the quantity $P_{\,\textnormal{PB}}^{(s)}[\phi|\mathcal{A}_{1+s\epsilon}(\hat{\rho})]$ the continuous probability distribution in the Pegg-Barnett formalism. However, as discussed extensively in Section \ref{sec:PB}, this is true only if in the formula for the expectation value (\ref{eq:PB_expectation_continuous_limit_general}) for the amplified state one can take the limit $s\to\infty$ under the integral, i.e.
\begin{align} \label{eq:lemma}
\begin{split}
    \lim_{s\to\infty}\braket{\hat{\phi}_{\textnormal{PB}}^{(s)}[f]}_{\mathcal{A}(1+s\epsilon,\,\hat{\rho})}\qquad\qquad\qquad\qquad\qquad\\
        = \int_{0}^{2\pi}d\phi \, f\left(e^{i\phi}\right) 
        \lim_{s\to\infty} P_{\textnormal{PB}}^{(s)}[\phi|\mathcal{A}_{1+s\epsilon}(\hat{\rho})].
\end{split}
\end{align}
Remarkably, we find that this always holds, even if it does not for the unamplified state. Thus, $P_{\,\textnormal{PB}}^{(s)}[\phi|\mathcal{A}_{1+s\epsilon}(\hat{\rho})]$ indeed constitutes the continuous probability distribution in the Pegg-Barnett formalism. This technical result is proved by us in Appendix \ref{app:integral_limit}.

We now proceed to prove our main result.

\begin{proof}[Proof of Proposition \ref{th:main_result}]
Setting $\kappa=1+s\epsilon$ in eq. (\ref{eq:QLA_standard_basis}) and making use of definitions (\ref{eq:PB_dist_continous}-\ref{eq:PB_phi_state}) leads to 
\begin{align} \label{eq:binomial_manipulation_before}
\begin{split}
    P_{\textnormal{PB}}^{(s)}[\phi|\mathcal{A}_{1+s\epsilon}(\hat{\rho})] 
        = \frac{1}{2\pi(1+s\epsilon)} \sum_{\substack{j=0}}^s
        \left(\frac{s\epsilon}{1+s\epsilon}\right)^{j}\\
        \times \sum_{\substack{m,n=0}}^{s-j}
        \rho_{mn}\frac{e^{i(n-m)\phi}}{\sqrt{(1+s\epsilon)}^{m+n}}
        \sqrt{\binom{j+m}{j}\binom{j+n}{j}},
\end{split}
\end{align}
where we note that the summation limits on $m,n,j$ follow from the fact that the Pegg-Barnett operator is limited to dimension $s$, which means that $j+m,j+n\in\{0,\ldots,s\}$. Our goal is to show that after taking the limits $s\to\infty$, $\epsilon\to 0$, in that order, the above quantity is equal to (\ref{eq:Paul_p}).

To simplify our considerations and shorten the notation, let us observe that eq. (\ref{eq:main_result_prob}) is linear in the density operator. For this reason, it is enough to restrict ourselves to pure states: $\hat{\rho}=\ket{\psi}\bra{\psi}$, for which  $\rho_{mn}=\psi_m\psi_n^*$. This assumption has no impact on the correctness of the proof. We get
\begin{align} \label{eq:binomial_manipulation_before_2}
\begin{split}
    P_{\textnormal{PB}}^{(s)}[\phi|\mathcal{A}_{1+s\epsilon}(\hat{\rho})]
        = \frac{1}{2\pi(1+s\epsilon)} \sum_{\substack{j=0}}^s
        \left(\frac{s\epsilon}{1+s\epsilon}\right)^{j}\\
        \times\left\lvert\sum_{\substack{m=0}}^{s-j}
        \psi_m\frac{e^{-im\phi}}{\sqrt{(1+s\epsilon)}^{m}}
        \sqrt{\binom{j+m}{j}}\right\rvert^2.
\end{split}
\end{align}
In the next step, we rewrite
\begin{align}
\begin{split}
    \frac{1}{(1+s\epsilon)^m}\binom{j+m}{j} = \frac{\prod_{k=1}^m(j+k)}{(1+s\epsilon)^m m!}
        = \frac{1}{m!}\prod_{k=1}^m \frac{j+k}{1+s\epsilon}.
\end{split}
\end{align}
Thus, eq. (\ref{eq:binomial_manipulation_before_2}) becomes
\begin{align} \label{eq:P_final_proof}
\begin{split}
    P_{\textnormal{PB}}^{(s)}[\phi|\mathcal{A}_{1+s\epsilon}(\hat{\rho})]
        = \frac{1}{2\pi(1+s\epsilon)} \sum_{\substack{j=0}}^s
        \left(\frac{s\epsilon}{1+s\epsilon}\right)^{j}\\
        \times\left\lvert\sum_{\substack{m=0}}^{s-j}
        \psi_m\frac{e^{-im\phi}}{\sqrt{m!}}
        \sqrt{\prod_{k=1}^m \frac{j+k}{1+s\epsilon}}\right\rvert^2.
\end{split}
\end{align}

At this point, it will be beneficial to turn the summation over $j$ into an integral. We do this in complete analogy to the case of particle in a box approaching
infinite volume \footnote{This is also how the integral formula for the expectation values
(\ref{eq:PB_expectation_continuous}) in the Pegg-Barnett formalism is obtained from its discrete counterpart (\ref{eq:PB_expectation}).}. Instead of summing over $j$ from $0$ to $s$, we sum over $\mu_j \coloneqq j/s$ from $0$ to $1$. In the limit of large $s$, in which we are interested in, the sum approaches an integral. As $\mu_j$ occupies the volume $1/s$ in the space of indices, we have
\begin{align}
\begin{split}
    j \to s\mu, \quad  \sum_{j=0}^s \to s \int_0^1 d\mu.
\end{split}
\end{align}
Therefore, for very large $s$
\begin{align} \label{eq:Z_summation_before_exchange}
\begin{split}
    P_{\textnormal{PB}}^{(s)}[\phi|\mathcal{A}_{1+s\epsilon}(\hat{\rho})]
        = \frac{s}{2\pi(1+s\epsilon)} 
        \int_0^1 d\mu \left(\frac{s\epsilon}{1+s\epsilon}\right)^{s\mu}\\
        \times\left\lvert\sum_{\substack{m=0}}^{s-s\mu}
        \psi_m\frac{e^{-im\phi}}{\sqrt{m!}}
        \sqrt{\prod_{k=1}^m \frac{s\mu+k}{1+s\epsilon}}\right\rvert^2.
\end{split}
\end{align}

We are now ready to take the limit $s\to\infty$. We can do this term by term, which is justified by the fact that each term has a well-defined limit. We have
\begin{align} \label{eq:taking_limit_s}
\begin{split} 
    \frac{s}{(1+s\epsilon)} \to \frac{1}{\epsilon}, \quad
    \left(\frac{s\epsilon}{1+s\epsilon}\right)^{s\mu} \to e^{-\mu/\epsilon}.
\end{split}
\end{align}
Finally, the bottom line of eq. (\ref{eq:Z_summation_before_exchange}), approaches
\begin{align} \label{eq:good_term}
\begin{split}
    \Bigg\lvert\sum_{\substack{m=0}}^{s-s\mu}
        \psi_m\frac{e^{-im\phi}}{\sqrt{m!}}
        &\sqrt{\prod_{k=1}^m \frac{s\mu+k}{1+s\epsilon}}\Bigg\rvert^2 \\
        &\to \left\lvert \sum_{\substack{m=0}}^{\infty}
        \psi_m\frac{e^{-im\phi}}{\sqrt{m!}}
        \sqrt{\frac{\mu}{\epsilon}}^m \right\rvert^2,
\end{split}
\end{align}
which we prove in Appendix \ref{app:bottom_line}.

This altogether yields
\begin{align} \label{eq:altogether}
\begin{split}
    \lim_{s\to\infty} 
    P_{\textnormal{PB}}^{(s)}[\phi|\mathcal{A}_{1+s\epsilon}(\hat{\rho})] =
        \frac{1}{2\pi\epsilon}
        \int_0^1 \frac{d\mu}{\epsilon} e^{-\mu/\epsilon}\\
        \times\left\lvert \sum_{\substack{m=0}}^{\infty}
        \psi_m\frac{e^{-im\phi}}{\sqrt{m!}}
        \sqrt{\frac{\mu}{\epsilon}}^m \right\rvert^2,
\end{split}
\end{align}
or, upon substituting $r^2 \coloneqq \mu/\epsilon$ and rearranging,
\begin{align} \label{eq:Z_recognize_Q}
\begin{split}
    \lim_{s\to\infty} P_{\textnormal{PB}}^{(s)}[\phi|\mathcal{A}_{1+s\epsilon}(\hat{\rho})] = \int_0^{\sqrt{1/\epsilon}} \frac{d r}{\pi} \, r \\
        \times\left\lvert e^{-r^2/2}\sum_{\substack{m=0}}^{\infty}
        \psi_m\frac{e^{-im\phi}}{\sqrt{m!}}r^m \right\rvert^2.
\end{split}
\end{align}
From the equation (\ref{eq:coherent_state}) for coherent state in the number basis and the definition (\ref{eq:Husimi_Q_distribution}) of the Husimi distribution one immediately recognizes that the bottom line equals $Q_{\hat{\rho}}(re^{i\phi})$. Thus,
\begin{align} \label{eq:Z_final}
\begin{split}
    \lim_{s\to\infty} P_{\textnormal{PB}}^{(s)}[\phi|\mathcal{A}_{1+s\epsilon}(\hat{\rho})] = \int_{0}^{\sqrt{1/\epsilon}} \frac{d r}{\pi} \, r \, Q_{\hat{\rho}}\left(re^{i\phi}\right),
\end{split}
\end{align}
which in the limit $\epsilon\to 0$ becomes the probability distribution (\ref{eq:Paul_p}) in the Paul formalism. This concludes the proof.
\end{proof}

As an immediate consequence of Proposition \ref{th:main_result}, the following corollary follows from the definitions of the expectation values in the two formalisms.

\begin{corollary} \label{th:corollary_expectations}
The expectation values (\ref{eq:Paul_expectation}) in the Paul formalism can be obtained from their Pegg-Barnett counterparts (\ref{eq:PB_expectation_continuous}) through the quantum limited amplifier as:
\begin{align} \label{eq:corollary_expectations}
    \braket{\hat{\phi}_{\textnormal{Paul}}[f]}_{\hat{\rho}}
        = \lim_{\epsilon\to 0}\lim_{s\to\infty}
        \braket{\hat{\phi}_{\textnormal{PB}}^{(s)}[f]}_{\mathcal{A}(1+s\epsilon,\,\hat{\rho})}.
\end{align}
\end{corollary}

\section{Discussion} \label{sec:discussion}
Let us discuss our results, beginning with their physical interpretation. As mentioned before, the amplification process is associated with making quantum phenomena more classical. Notably, it is known to transform the Glauber P distribution into the more semi-classical Husimi Q distribution \cite{quantum_phase_Q_amplification_Schleich_1992}, and the von Neumann entropy into the more classical-like Wehrl entropy \cite{De_Palma_2017}. In the view of our work, this suggests that the Paul formalism may be viewed a semi-classical limit of the Pegg-Barnett formalism.

This interpretation is strengthened by the fact that, while for a generic quantum state the Pegg-Barnett probability distribution may not exist in the infinite dimension, it does for all amplified states, as if removing all the quantum ``singularities''. Note also that the Paul formalism, to which the amplification leads from the Pegg-Barnett framework, is itself invariant under state amplification: 
\begin{equation}
    P_{\,\textnormal{Paul}}\left[\phi|\mathcal{A}_\kappa(\hat{\rho})\right]
        =P_{\,\textnormal{Paul}}(\phi|\hat{\rho}) \textnormal{ for all } \kappa\geqslant 1.
\end{equation}
To see this, one needs to make use of the known relation \cite{De_Palma_2017}
\begin{equation}
    Q_{\mathcal{A}(\kappa,\,\hat{\rho})}(\alpha)=\frac{1}{\kappa}\, Q_{\hat{\rho}}\left(\frac{\alpha}{\sqrt{\kappa}}\right)
\end{equation}
in eq. (\ref{eq:Paul_p}) and change the integration variable to $r'=\sqrt{\kappa}r$. Thus, if we consider the Paul formalism to be the Pegg-Barnett formalism with some of its quantum features suppressed through infinite amplification, it is only natural that further amplification leaves it unaffected.

Why do our results assume the specific amplification choice $\kappa=1+s\epsilon$? In particular, why do they connect the amplification rate to the Pegg-Barnett dimension in a linear way? From the mathematical point of view, the necessity of such connection is clear: if we repeat the derivation of our main results with amplification strength $\kappa$ set to be either independent of $s$ or dependent on it in a non-linear way (e.g. $\kappa=1+s^2\epsilon$), we would quickly find the Pegg-Barnett probability to be vanishing in the limit of infinite amplification. See Appendix \ref{app:linearity}, where we show this explicitly. Thus, setting $\kappa$ to be linear in $s$, as we did, is necessary to obtain a non-trivial limit. This also shows that a similar result to ours cannot hold in the Susskind-Glogower formalism, since there $s=\infty$ from the beginning, making it impossible to set $\kappa$ dependent on $s$.

As a final remark, we observe that Corollary \ref{th:corollary_expectations} can be alternatively formulated as a relation between the phase operators in the two formalisms. Let us deploy the quantum limited attenuator channel, whose action on arbitrary operator $\hat{O}$ reads \cite{De_Palma_2017}
\begin{align}
    \mathcal{E}_\lambda\big(\hat{O}\big)\coloneqq
        \Tr_B\left[\hat{V}_\lambda\big(\hat{O}\otimes\ket{0}\bra{0}\big)
        \hat{V}^\dag_\lambda\right].
\end{align}
Here, $0\leqslant \lambda \leqslant 1$ (where $\lambda=1$ corresponds to the identity channel) and
\begin{align}
    \hat{V}_\lambda\coloneqq
        \exp\left[\arccos\sqrt{\lambda}(\hat{a}^\dag\hat{b}-\hat{a}\hat{b}^\dag)\right].
\end{align}
The quantum limited attenuator is dual to the quantum limited amplifier, by which we mean that for any state $\hat{\rho}$, operator $\hat{O}$ and $\kappa\geqslant 1$ we have
\begin{align}
\begin{split}
    \Tr \mathcal{A}_\kappa(\hat{\rho})\, \hat{O} 
        = \Tr \hat{\rho}\,\frac{1}{\kappa}\mathcal{E}_{1/\kappa}\big(\hat{O}\big).
\end{split}
\end{align}
Therefore, infinite amplification of the state $\kappa\to\infty$ is equivalent to infinite attenuation $\lambda=1/\kappa\to 0$ of the operator.

Applying this to Proposition \ref{th:main_result} we find that for any $\hat{\rho}$ and $f$:
\begin{align} \label{eq:main_result_operator}
    \Tr \hat{\rho}\left[\hat{\phi}_{\textnormal{Paul}}[f]
        -\lim_{\epsilon\to 0}\lim_{s\to \infty}\frac{1}{1+s\epsilon}\mathcal{E}_{1/(1+s\epsilon)}
        \left(\hat{\phi}_{\textnormal{PB}}^{(s)}[f]\right)\right]=0.
\end{align}
Because this equation holds for arbitrary input state $\hat{\rho}$, it is tempting to say that the Paul operator is equal to the infinitely-attenuated Pegg-Barnett operator, i.e.
\begin{align} \label{eq:main_result_operator_explicit}
    \hat{\phi}_{\textnormal{Paul}}[f]
        =\lim_{\epsilon\to 0}\lim_{s\to \infty}\frac{1}{1+s\epsilon}\mathcal{E}_{1/(1+s\epsilon)}
        \left(\hat{\phi}_{\textnormal{PB}}^{(s)}[f]\right).
\end{align}
However, we only showed that the two operators coincide when traced with a formal density operator, i.e. a non-negative, hermitian operator. Therefore, while the Paul and infinitely attenuated Pegg-Barnett operators are clearly connected, whether they are completely equal, as in eq. (\ref{eq:main_result_operator_explicit}), remains to be proven (or disproven). 

In any case, because of the general postulate of the Pegg-Barnett formalism to calculate the expectation values first and only then take the limit of infinite dimension, one has to be careful with such an operator interpretation, as an infinitely-dimensional phase operator is technically not a part of the Pegg-Barnett framework.

\section{Examples} \label{sec:examples}
We illustrate our with a number of examples. We begin with the simple case of thermal states, for which the convergence of the two phase formalisms is easy to see explicitly.

\begin{figure}[!t]
    \centering
    \begin{subfigure}[b]{0.49\textwidth}
         \centering
         \includegraphics[width=\textwidth]{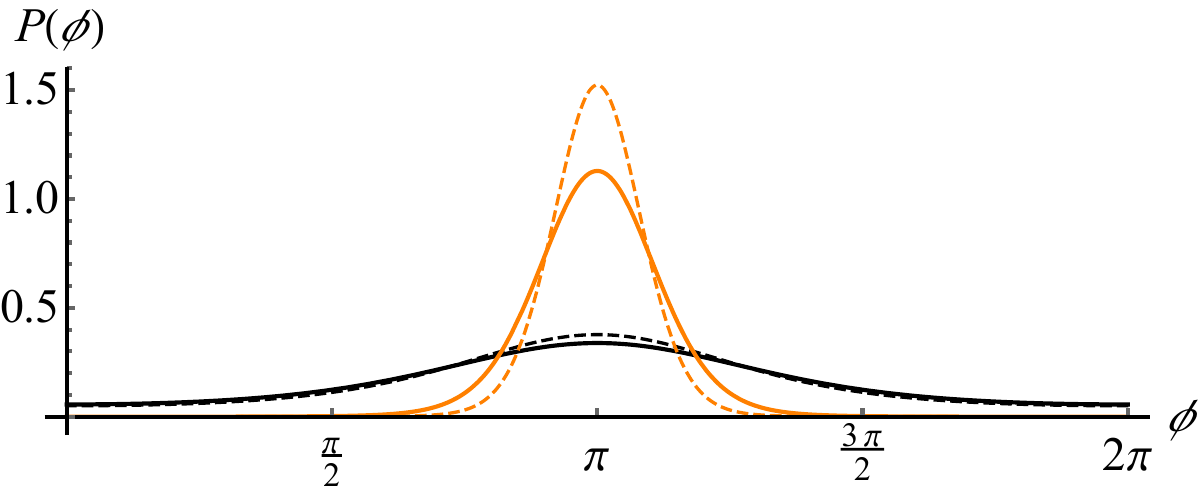}
        \subcaption[]{} \label{fig:a}
     \end{subfigure} 
      \begin{subfigure}[b]{0.49\textwidth}
         \centering
         \vspace{0pt}
         \includegraphics[width=\textwidth]{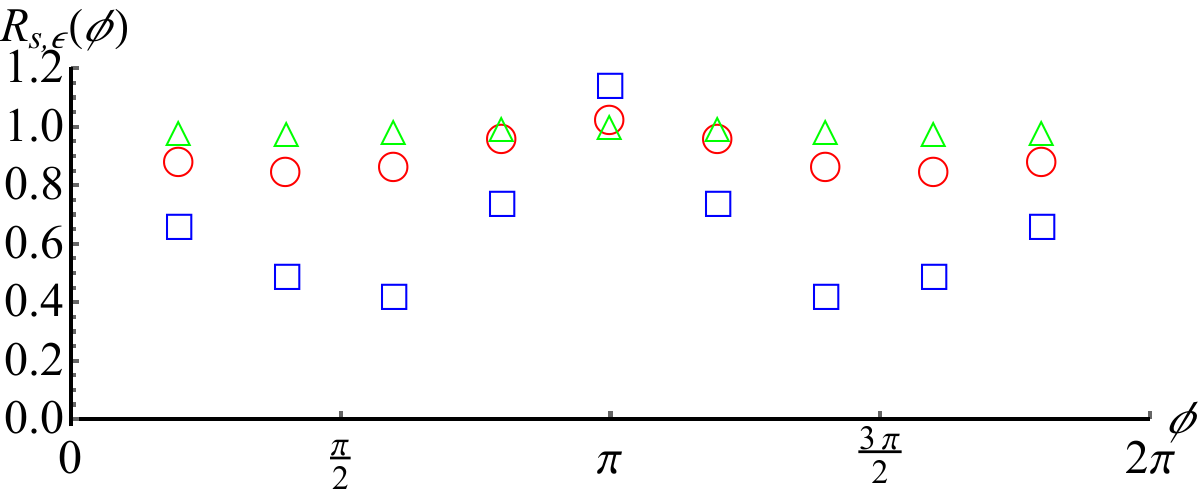}
        \subcaption[]{} \label{fig:b}
     \end{subfigure}
   \caption{a) Comparison between the Paul phase distribution $P_{\,\textnormal{Paul}}(\phi|\hat{\rho})$ (thick lines) and its Pegg-Barnett counterpart $P_{\textnormal{PB}}(\phi|\hat\rho)$ (dashed lines) for $\psi=\pi$ and $r'=0.5$ (black), $r'=2$ (orange). To calculate $P_{\textnormal{PB}}(\phi|\hat{\rho})$, the sum in eq. (\ref{eq:P(phi)c}) was approximated by its first 100 terms. b) Point plot of the ratio $R_{s,\epsilon}(\phi)$ for the coherent state given by $r'=2$ and $\psi=\pi$ calculated at the points $\phi=2\pi t/10$ with $t\in\{1,..,9\}$ for $\epsilon=0.01$. Blue squares, red circles and green triangles stand for $s+1\in\{10^2,10^3,10^4\}$, respectively. As expected, for large $s$ the ratio approaches the value of one.}
\label{fig:coh}
\end{figure}

\begin{example}
Let us consider the thermal states of the harmonic oscillator:
\begin{align} \label{eq:thermal_state_general}
    \hat{g}_\beta \coloneqq \frac{e^{-\beta\hat{a}^\dag \hat{a}}}{\Tr e^{-\beta\hat{a}^\dag \hat{a}}},
\end{align}
where $\beta>0$. For such states, the Paul probability distribution (\ref{eq:Paul_p}) can be calculated analytically, yielding the flat distribution:
\begin{align} \label{eq:Paul_expectation_thermal_example}
    P_{\,\textnormal{Paul}}(\phi|\hat{g}_\beta)
        = \frac{1}{2\pi}.
\end{align}
To compare with the Pegg-Barnett formalism, we begin with eq. (\ref{eq:QLA_standard_basis}), finding that
\begin{align}
    \mathcal{A}_\kappa(\hat{g}_\beta) = \hat{g}_{\beta(\kappa)},
        \quad \beta(\kappa) = \ln\frac{\kappa}{e^{-\beta}+\kappa-1}.
\end{align}
Substituting this into eq. (\ref{eq:binomial_manipulation_before_2}) and simplifying yields
\begin{align}
    P^{(s)}_{\,\textnormal{PB}}(\phi|\hat{g}_{\beta(1+s\epsilon)}) = \frac{1}{2\pi} 
        \left[1-\left(\frac{e^{-\beta}+s\epsilon}{1+s\epsilon}\right)^{s+1}\right].
\end{align}
Taking the limit $s\to\infty$ we get
\begin{align}
    \lim_{s\to\infty}P^{(s)}_{\,\textnormal{PB}}(\phi|\hat{g}_{\beta(1+s\epsilon)}) 
        = \frac{1}{2\pi}\left[1-e^{-\left(1-e^{-\beta}\right)/\epsilon}\right],
\end{align}
which clearly coincides with the Paul probability distribution (\ref{eq:Paul_expectation_thermal_example}) after taking the limit $\epsilon\to 0$.
\end{example}

In the remaining examples, we consider states with non-trivial phase dependence. To study the convergence of the two formalisms, we employ numerical methods. To this end, it is useful to define the following object
\begin{align} \label{eq:R}
    R_{s,\epsilon}(\phi)
        \coloneqq \frac{P_{\,\textnormal{PB}}^{(s)}
        [\phi|\mathcal{A}_{1+s\epsilon}(\hat{\rho})]}
        {P_{\,\textnormal{Paul}}(\phi|\hat{\rho})},
\end{align}
which is simply the ratio of the Pegg-Barnett ``amplified'' probability distribution to the Paul probability distribution. According to Proposition \ref{th:main_result}, this ratio should approach the value of one for large $s$ and small $\epsilon$.

\begin{example}
As a second example, let us consider a coherent state $\hat{\rho}_\alpha=\ket{\alpha}\bra{\alpha}$ with amplitude $\alpha = r'e^{i\psi}$. This example is especially relevant for the study of quantum phase, since the phase of such a coherent state is approximately equal to $\psi$.

In this case, the Husimi Q distribution equals
\begin{equation}
    Q_{\hat{\rho}_\alpha}(re^{i\phi})=e^{-r^2-r'^2+2rr'\cos(\phi-\psi)},
\end{equation}
resulting in the following Paul phase distribution:
\begin{multline}
    P_{\,\textnormal{Paul}}(\phi|\hat{\rho}_\alpha)=\frac{e^{-r'^2} }{2 \pi }\left(1+\sqrt{\pi } r' \cos (\phi -\psi ) e^{r'^2 \cos ^2(\phi -\psi )}\right. \\
   \times \left[\text{erf}\left[r' \cos (\phi -\psi )\right]+1\right]\Big),
\end{multline}
where $\text{erf}$ stands for the error function. On the other hand, calculating from definition, the corresponding Pegg-Barnett phase distribution reads
\begin{equation}
    P_{\,\textnormal{PB}}(\phi|\hat{\rho}_\alpha)=\frac{e^{-r'^2}}{2 \pi } \left| \sum _{n=0}^{\infty } \frac{e^{i n (\psi -\phi )} r'^n}{\sqrt{n!}}\right|^2.\label{eq:P(phi)c}
\end{equation}
Its ``amplified'' version, by which we mean here $P_{\textnormal{PB}}^{(s)}[\phi|\mathcal{A}_{1+s\epsilon}(\hat{\rho}_\alpha)]$, follows readily by substituting $\psi_{m}=\alpha_{m}$, with $\alpha_n$ as in eq. (\ref{eq:coherent_state}), into eq. (\ref{eq:P_final_proof}).

We compare the two distributions in Figure \ref{fig:a}. As expected from coherent states, both phase distributions peak at $\psi=\pi$, with the effect being more pronounced for larger $r'$. While the Paul and Pegg-Barnett frameworks both give the same qualitative results, the Paul framework yields a noticeably less pronounced peak. Nonetheless, the Paul distribution can be obtained from the Pegg-Barnett distribution by using the quantum limited amplifier, as seen from Figure \ref{fig:b}.
\end{example}

\begin{table*}[!t]
\centering
\begin{tabularx}{0.98\textwidth}{|X c|X X X X X|} 
 \hline
 \multicolumn{2}{|X|}{\multirow{2}{*}{$\qquad\qquad\:\: R_{s,\epsilon}(0.3)$}} & \multicolumn{5}{c|}{$s$} \\
 \multicolumn{2}{|X|}{\multirow{2}{*}{}}  & $10^0$ & $10^1$ & $10^2$ & $10^3$ & $10^4$ \\
 \hline
 \multirow{5}{*}{$\qquad\qquad\epsilon$} 
 & $1.00\quad$
 & $0.61 \pm 0.54$
 & $0.47 \pm 0.37$
 & $0.45 \pm 0.43$
 & $0.45 \pm 0.44$
 & $0.45 \pm 0.44$
 \\ 
 & $0.50\quad$
 & $1.19 \pm 1.01$
 & $0.80 \pm 0.44$
 & $0.73 \pm 0.31$
 & $0.72 \pm 0.30$
 & $0.72 \pm 0.30$
 \\
 & $0.10\quad$
 & $5.27 \pm 4.60$
 & $1.79 \pm 0.95$
 & $1.08 \pm 0.14$
 & $1.01 \pm 0.02$
 & $1.00 \pm 0.00$
 \\
 & $0.05\quad$
 & $10.17 \pm 9.08$
 & $2.67 \pm 1.80$
 & $1.18 \pm 0.24$
 & $1.02 \pm 0.03$
 & $1.00 \pm 0.00$
 \\
 & $0.01\quad$
 & $49.13 \pm 44.92$
 & $9.75 \pm 8.34$
 & $1.95 \pm 1.04$
 & $1.09 \pm 0.14$
 & $1.01 \pm 0.02$
 \\
 \hline\hline
\end{tabularx}
\caption{Numerical values of the ratio $R_{s,\epsilon}(\phi)$ [defined in eq. (\ref{eq:R})] of the Pegg-Barnett probability distribution calculated on an amplified state to the Paul probability distribution for $\phi=0.3$. Each entry is an average over 1000 random qubit density matrices (i.e. single-photon states) sampled from the Hilbert-Schmidt ensemble, with terms after $\pm$ standing for maximum deviation from the mean value. All values are rounded to two significant digits. As seen, the ratio approaches the value $R=1$ in the limit of growing $s$ and vanishing $\epsilon$, as long as $s$ is much bigger than $1/\epsilon$, which we interpret as taking the limit $s\to\infty$ before taking the limit $\epsilon\to 0$.}
\label{tab:numerics}
\end{table*}

\begin{example}
The purpose of the final example is to test how the limiting procedure works in practice. In Table \ref{tab:numerics}, we provide approximate numerical values of the ratio (\ref{eq:R})
for various values of $s$ and $\epsilon$ numerically averaged over random qubit density matrices, i.e. single-photon states, sampled from the Hilbert-Schmidt ensemble \cite{random_density_matrices_Zyczkowski_2011}. As seen, the ratio approaches the value of one in the limit $s\to\infty$, $\epsilon\to 0$, provided $s$ is much bigger than $1/\epsilon$, which we interpret as taking the limit $s\to\infty$ before taking the limit $\epsilon\to 0$.
\end{example}

\section{Concluding remarks}
\label{sec:summary}
To shortly conclude, we successfully demonstrated that the Paul formalism can be obtained from the Pegg-Barnett formalism intertwined with infinite amplification of the state. Our findings suggest a number of closely related directions for future research. Firstly, the Pegg-Barnett framework is known to coincide with the Paul framework for small-amplitude coherent states, and with the Paul framework's analogue based on the Wigner distribution for large-amplitude coherent states \cite{quasiprobability_phase_Tanas_1996}. Given the former association's loose resemblance of our main result, perhaps there exists a physical process akin to amplification, which connects the Pegg-Barnett operator to the Wigner distribution? Secondly, the phase-difference operator \cite{quantum_phase_difference_operator_Luis_1993} is based on a similar construction to the Pegg-Barnett operator. Therefore, one may expect that it is also connected to the Paul formalism through some type of system amplification. Finally, our work, as well as the previously mentioned Refs \cite{quantum_phase_Q_amplification_Schleich_1992,De_Palma_2017}, suggest that further ``quantum to classical'' transitions are potentially possible by virtue of the amplification procedure.

\begin{acknowledgements}
We thank S. Cusumano, A. Mandarino, A. Luk\v{s} and J. K\v{r}epelka for their remarks, as well as M. Paris for recommending Ref. \cite{Paul_formalism_QLA_Lalovi_1998} to us. We acknowledge support by the Foundation for Polish Science (International Research Agenda Programme project, International Centre for Theory of Quantum Technologies, Grant No. 2018/MAB/5, cofinanced by the European Union within the Smart Growth Operational Program). This work is partially carried out under International Research Agenda Programme, project no. FENG.02.01- IP.05-0006/23, financed by the FENG program 2021-2027, Priority FENG.02, Measure FENG.02.01., with the support of the Foundation for Polish Science.
\end{acknowledgements}

\bibliography{report}
\bibliographystyle{obib}

\appendix
\section{Proof of eq. (\ref{eq:lemma})} 
\label{app:integral_limit}
\setcounter{equation}{0}
\renewcommand{\theequation}{\ref{app:integral_limit}\arabic{equation}}
In this appendix, we prove eq. (\ref{eq:lemma}), meaning that in the formula for the expectation value (\ref{eq:PB_expectation_continuous_limit_general}) calculated on an amplified state the limiting procedure and the integration can be swapped. In other words, we prove that
\begin{align} \label{eq:lemma_appendix}
\begin{split}
    \lim_{s\to\infty} \int_{0}^{2\pi}d\phi \,  I_s(\phi) =
        \int_{0}^{2\pi}d\phi \, 
        \lim_{s\to\infty} I_s(\phi),
\end{split}
\end{align}
where we denoted for short
\begin{align}
    I_s(\phi) \coloneqq   f\left(e^{i\phi}\right)
        P_{\textnormal{PB}}^{(s)}[\phi|\mathcal{A}_{1+s\epsilon}(\hat{\rho})].
\end{align}

According to the dominated convergence theorem, a sufficient condition for eq. (\ref{eq:lemma_appendix}) to hold is if there exists an $s$-independent function $J(\phi)$, such that
\begin{align} \label{eq:J_integral}
    \int_{0}^{2\pi}d\phi\,J(\phi)<\infty
\end{align}
and for all $s$, $\phi$
\begin{align} \label{eq:I_J_relation}
    |I_s(\phi)|\leqslant J(\phi).
\end{align}
We make the following guess:
\begin{align}
\begin{split}
    J(\phi) = \frac{\max_x |f(x)|}{2\pi\epsilon} 
        \left(\sum_{\substack{m=0}}^{\infty}
        \frac{1}{\sqrt{m! \epsilon^m}}\right)^2.
\end{split}
\end{align}
We stress that in Proposition \ref{th:main_result}, the limit $s\to\infty$ is taken before the limit $\epsilon\to 0$, meaning that $J(\phi)$ is finite. Condition (\ref{eq:J_integral}) is thus obviously fulfilled. It remains to show eq. (\ref{eq:I_J_relation}).

To this end, we start with the expression (\ref{eq:P_final_proof}). Using the fact that $s/(1+s\epsilon)\leqslant 1/\epsilon$ and $\left(\frac{s\epsilon}{1+s\epsilon}\right)^{s\mu}\leqslant 1$, as well as basic properties of absolute value, we get
\begin{align}
\begin{split}
    \left| I_s(\phi)\right|
        \leqslant \frac{ \left| f\left(e^{i\phi}\right)\right|}{2\pi\epsilon} 
        \int_0^1 d\mu \left(\sum_{\substack{m=0}}^{s-s\mu}
        \frac{1}{\sqrt{m!}}
        \sqrt{\prod_{k=1}^m \frac{s\mu+k}{1+s\epsilon}}\right)^2.
\end{split}
\end{align}
We then observe that under the product we have $k\leqslant m \leqslant s-s\mu$, which yields
\begin{align}
\begin{split}
    \left|I_s(\phi)\right|
        \leqslant \frac{\left| f\left(e^{i\phi}\right)\right|}{2\pi\epsilon} 
        \int_0^1 d\mu \left(\sum_{\substack{m=0}}^{s-s\mu}
        \frac{1}{\sqrt{m!}}
        \sqrt{\frac{s}{1+s\epsilon}}^m\right)^2.
\end{split}
\end{align}
In the last step, we once again use $s/(1+s\epsilon)\leqslant 1/\epsilon$. Furthermore, because all the summands are non-negative, we extend the sum to infinity. Finally, we can perform the integral over $\mu$. In the end, we have
\begin{align}
\begin{split}
    \left|I_s(\phi)\right|
        \leqslant \frac{\left| f\left(e^{i\phi}\right)\right|}{2\pi\epsilon} 
        \left(\sum_{\substack{m=0}}^{\infty}
        \frac{1}{\sqrt{m! \epsilon^m}}\right)^2.
\end{split}
\end{align}
Substituting this into the l.h.s. of eq. (\ref{eq:I_J_relation}) and bounding $|f|$ from above by its largest value finishes the proof.

Note that this result also shows that one can interchange the limit $s\to\infty$ with the integral over $\mu$ in eq. (\ref{eq:Z_summation_before_exchange}).

\section{Proof of eq. (\ref{eq:good_term})} 
\label{app:bottom_line}
\setcounter{equation}{0}
\renewcommand{\theequation}{\ref{app:bottom_line}\arabic{equation}}
In this appendix, we want to show that the limit $s\to\infty$ of the bottom line of eq. (\ref{eq:Z_summation_before_exchange}) is given by eq. (\ref{eq:good_term}).

Let us denote
\begin{align}
\begin{split}
    S_{x,y} = \sum_{\substack{m=x}}^{y}
        \psi_m\frac{e^{-im\phi}}{\sqrt{m!}}
        \sqrt{\prod_{k=1}^m \frac{s\mu+k}{1+s\epsilon}},
\end{split}
\end{align}
so that the bottom line of eq. (\ref{eq:Z_summation_before_exchange}) equals $\left\lvert S_{0,s-s\mu}\right\rvert^2$. We now introduce an auxiliary parameter $d\in\mathbb{N}$, $d\leqslant s-s\mu$, and split the sum over $m$ into two sums: one from $0$ to $d-1$, and the second one from $d$ to $s-s\mu$. We get
\begin{align}
\begin{split}
    \left\lvert S_{0,s-s\mu}\right\rvert^2 = \left\lvert S_{0,d-1} + S_{d,s-s\mu}\right\rvert^2.
\end{split}
\end{align}
We stress that $S_{x,y}$ is independent of $d$.

Using basic properties of the absolute value, we get the following bounds:
\begin{align} \label{eq:bottom_line_bounds}
\begin{split}
    \left\lvert S_{0,d-1}\right\rvert^2 - \left\lvert S_{d,s-s\mu}\right\rvert^2
    \leqslant \left\lvert S_{0,s-s\mu}\right\rvert^2 
    \leqslant \left\lvert S_{0,d-1}\right\rvert^2 + \left\lvert S_{d,s-s\mu}\right\rvert^2.
\end{split}
\end{align}
Our approach is to calculate the limit $s \rightarrow \infty$ separately for the two terms present in the bounds, ultimately showing that both bounds coincide and are therefore equal to the limit of $\left\lvert S_{0,s-s\mu}\right\rvert^2$.

Because $|\psi_m|\leqslant 1$ and $\left\lvert e^{-im\phi}\right\rvert=1$ we can see that
\begin{align}
\begin{split}
    \left\lvert S_{d,s-s\mu}\right\rvert \leqslant \sum_{\substack{m=d}}^{s-s\mu}
        \frac{1}{\sqrt{m!}}\sqrt{\prod_{k=1}^m \frac{s\mu+k}{1+s\epsilon}}.
\end{split}
\end{align}
Furthermore, $k$ is bounded from above by $m$, which is in turn bounded by $s-s\mu$. Thus
\begin{align} \label{eq:garbage_term}
\begin{split}
    \left\lvert S_{d,s-s\mu}\right\rvert & \leqslant \sum_{\substack{m=d}}^{s-s\mu}
        \frac{1}{\sqrt{m!}}\sqrt{\prod_{k=1}^{m} \frac{s}{1+s\epsilon}} \\
        & \to \sum_{\substack{m=d}}^{\infty}
        \frac{1}{\sqrt{m!}}\sqrt{\frac{1}{\epsilon}}^m = \sum_{\substack{m=0}}^{\infty}
        \frac{1}{\sqrt{(d+m)!}}\sqrt{\frac{1}{\epsilon}}^{d+m},
\end{split}
\end{align}
where in the second transition we performed the limit $s\to\infty$ and in the third (final) transition we renumbered the sum.

In the case of $S_{0,d-1}$, $m$ is bounded from above by the finite number $d$, which means that $k$ yields no contribution in the limit of infinite $s$ and so
\begin{align} \label{eq:good_term_nearly}
\begin{split}
    \left\lvert S_{0,d-1} \right\rvert^2 \to \left\lvert \sum_{\substack{m=0}}^{d-1}
        \psi_m\frac{e^{-im\phi}}{\sqrt{m!}}
        \sqrt{\frac{\mu}{\epsilon}}^m \right\rvert^2.
\end{split}
\end{align}

Since $d$ was chosen as an arbitrary number smaller than $s-s\mu$ and, as already pointed out,  $S_{0,s-s\mu}$ is independent of $d$, after taking the limit $s\to\infty$, we can pick up whatever value of $d\in\mathbb{N}$ we want to. In particular, we can now also take the limit $d\to\infty$. Looking at eq. (\ref{eq:garbage_term}), we can see that in this limit $S_{d,s-s\mu}$ vanishes -- because this series is absolutely convergent we can interchange the sum with the limit.
Thus, due to eq. (\ref{eq:bottom_line_bounds}), in the limit of infinite $s$, the bottom line of eq. (\ref{eq:Z_summation_before_exchange}) coincides with eq. (\ref{eq:good_term_nearly}) with $d\to\infty$. This proves eq. (\ref{eq:good_term}).

\section{Discussion of amplification rate} 
\label{app:linearity}
\setcounter{equation}{0}
\renewcommand{\theequation}{\ref{app:linearity}\arabic{equation}}
In this appendix, we show why Proposition \ref{th:main_result} no longer holds if the amplification rate is set to be non-linear in $s$, as opposed to the linear dependence $\kappa = 1 + s\epsilon$.

First, let us briefly discuss the case in which $\kappa$ is independent of $s$. In this case, we find that the analogue of eq. (\ref{eq:binomial_manipulation_before}) reads
\begin{align}
\begin{split}
    P_{\textnormal{PB}}^{(s)}[\phi|\mathcal{A}_{\kappa}(\hat{\rho})] 
        = \frac{1}{2\pi \kappa} \sum_{\substack{j=0}}^s
        \left(\frac{\kappa-1}{\kappa}\right)^{j}\\
        \times \sum_{\substack{m,n=0}}^{s-j}
        \rho_{mn}\frac{e^{i(n-m)\phi}}{\sqrt{\kappa}^{m+n}}
        \sqrt{\binom{j+m}{j}\binom{j+n}{j}},
\end{split}
\end{align}
Taking the limit of infinite amplification, $\kappa\to \infty$, we get simply zero.

To see what happens for $\kappa$ depending on $s$ in a non-linear way, let us consider $\kappa = 1 + w(s,\epsilon)$, where $w$ is a polynomial in $s$ and $\epsilon$ of at least second degree in $s$. In this case, by following exactly the same steps as in the original derivation, we find that the analogue of eq. (\ref{eq:Z_summation_before_exchange}) reads
\begin{align}
\begin{split}
    P_{\textnormal{PB}}^{(s)}[\phi|\mathcal{A}_{1+w(s,\epsilon)}(\hat{\rho})]
        = \frac{(2\pi)^{-1}s}{[1+w(s,\epsilon)]} 
        \int_0^1 d\mu \left(\frac{w(s,\epsilon)}{1+w(s,\epsilon)}\right)^{s\mu}\\
        \times\left\lvert\sum_{\substack{m=0}}^{s-s\mu}
        \psi_m\frac{e^{-im\phi}}{\sqrt{m!}}
        \sqrt{\prod_{k=1}^m \frac{s\mu+k}{1+w(s,\epsilon)}}\right\rvert^2.
\end{split}
\end{align}
By construction, $\lim_{s\to\infty} s/w(s,\epsilon) = 0$. From this, it is easy to see that the whole equation vanishes in the limit $s\to\infty$.

\end{document}